\documentclass[paper]{ieice}
\usepackage[dvipdfmx]{graphicx,xcolor}
\usepackage[fleqn]{amsmath}
\usepackage{amsthm} 
\usepackage{newtxtext}
\usepackage[varg]{newtxmath}
\usepackage{tikz} 
\usepackage{cite}
\usepackage{enumitem}

\field{D}
\title{Non-Cooperative Rational Synthesis Problem for Probabilistic Strategies}
\titlenote{%
  This work was partially supported by JSPS KAKENHI Grant Numbers 23K24824 and 24K14901.}
\authorlist{%
 \authorentry{So Koide}{n}{nagoya}%
 \authorentry[takata.yoshiaki@kochi-tech.ac.jp]{Yoshiaki TAKATA}{m}{kochi}%
 \MembershipNumber{9618813}
 \authorentry[seki@i.nagoya-u.ac.jp]{Hiroyuki Seki}{f}{nagoya}%
 \MembershipNumber{8206514}   
}
\affiliate[nagoya]{The authors are with the
 Graduate School of Informatics,
 Nagoya University,
 Furo-cho, Chikusa, Nagoya 464-8601, Japan.}
\affiliate[kochi]{The author is with the
 School of Informatics,
 Kochi University of Technology,
 Tosayamada, Kami City, Kochi 782-8502, Japan.}

\received{2015}{1}{1}
\revised{2015}{1}{1}

\newtheorem{theorem}{Theorem}
\newtheorem{proposition}[theorem]{Proposition}
\newtheorem{lemma}[theorem]{Lemma}
\newtheorem{corollary}[theorem]{Corollary}
\theoremstyle{definition}
\newtheorem{example}[theorem]{Example}

\newcommand{\widerbar}[1]{%
  \ensuremath{\text{
    \rlap{$\mskip2mu\overline{\scalebox{0.8}[1]{$\phantom{#1}$}}$}}%
  #1}}
\newcommand{\barsigma}{\widerbar{\sigma}}
\newcommand{\baralpha}{\widerbar{\alpha}}
\newcommand{\barz}{\widerbar{z}}
\newcommand{\barr}{\widerbar{r}}

\newcommand{\psisys}{\psi^{\mathit{sys}}}
\newcommand{\psienv}{\psi^{\mathit{env}}}
\newcommand{\Deltapd}{\Delta_{\smash{\mathit{pd}}}}

\newcommand{\Nat}{\mathbb{N}}

\newcommand{\calG}{\mathcal{G}}
\newcommand{\calA}{\mathcal{A}}
\newcommand{\Nature}{\mathit{nature}}
\newcommand{\Inf}{\mathit{inf}}
\newcommand{\Next}{\mathit{next}}
\newcommand{\nextPr}{\mathit{nextPr}}
\newcommand{\PR}{\mathit{Pr}}
\newcommand{\Pay}{\mathit{Pay}}
\newcommand{\NE}{\mathit{NE}}
\newcommand{\ZNE}{0\mathit{NE}}
\newcommand{\ZpNE}{0'\mathit{NE}}

\newcommand{\SUPP}{\mathrm{SUPP}}
\newcommand{\BSCC}{\mathrm{BSCC}}

\newcommand{\COND}{\mathit{COND}}

\newcommand{\EC}{\mathrm{EC}}
\newcommand{\OmegaX}[1]{\Omega_{-{#1}}}
\newcommand{\TR}{\mathit{TR}}
\newcommand{\Buchi}{\ifmmode\text{\textit{B\"{u}chi}}\else B\"{u}chi\fi}

\newcommand{\Muller}{\mathit{Muller}}

\newcommand{\OC}[1]{\mathbf{#1}}

\usetikzlibrary{shapes.geometric,positioning,calc}

\begin{document}
\maketitle
\begin{summary}
We study the decidability and complexity of 
non-cooperative rational synthesis problem (abbreviated as NCRSP) 
for some classes of probabilistic strategies.  
We show that NCRSP for stationary strategies and Muller objectives 
is in 3-EXPTIME, 
and if we restrict the strategies of environment players to be positional, 
NCRSP becomes NEXPSPACE solvable. 
On the other hand, NCRSP$_{>}$, which is a variant of NCRSP, is 
shown to be undecidable even for pure finite-state strategies and 
terminal reachability objectives. 
Finally, we show that NCRSP becomes EXPTIME solvable if we 
restrict the memory of a strategy to be the most recently visited
$t$ vertices where $t$ is linear in the size of the game. 
\end{summary}

\section{Introduction}

Synthesis problems on multiplayer non-zero-sum games (abbreviated as MG) 
between the system and multiple environment players
are the problems to find a good strategy of the system 
and have been extensively studied. 
An MG is generally stochastic in the sense that 
a special controller other than players, called \emph{nature},
which chooses a move in its turn randomly, may exist and 
a strategy that each player takes can be probabilistic. 
If there is no nature and every strategy is deterministic, 
the game is called non-stochastic.  

Two types of synthesis problems exist: 
cooperative rational synthesis problem (CRSP) and 
non-cooperative rational synthesis problem (NCRSP)\@.
The rationality of environment players is modeled by Nash equilibria, and
CRSP is the problem to decide whether there exists
a Nash equilibrium that gives the system a payoff not less than
a given threshold.
Ummels et~al.\ \cite{UW11} studied the complexity of CRSP for various classes of
objectives and strategies of players.
CRSP fits the situation where the system can make a suggestion of a strategy
profile (a tuple of strategies of all players) to the environment players.
However, in real applications, the system may rarely have an opportunity 
to make suggestions to the environment, and thus CRSP is optimistic.
NCRSP is 
the problem to decide whether there exists a strategy $\sigma_0$
of the system satisfying that
for \emph{every} strategy profile of the environment players
that forms a 0-fixed Nash equilibrium (a Nash equilibrium where
the system's strategy is fixed to $\sigma_0$),
the system obtains a payoff not less than a given threshold.

Condurache et al.~\cite{CFGR16} have investigated the complexity of CRSP and NCRSP
for {\em non-stochastic} MG. 
In \cite{KTS24}, the authors studied the complexity of NCRSP for {\em stochastic} MG, 
but the strategies are restricted to be positional, i.e., 
every strategy has no memory and is deterministic. 

This paper investigates the decidability and complexity of NCRSP
for some classes of strategies broader than positional ones. 
We show that NCRSP for stationary strategies and Muller objectives 
is in 3-EXPTIME, 
and if we restrict the strategies of environment players to be positional, 
NCRSP becomes NEXPSPACE solvable. 
On the other hand, NCRSP$_{>}$, which is a variant of NCRSP, is 
shown to be undecidable even for pure finite-state strategies and 
terminal reachability objectives. 
Finally, we show that NCRSP becomes EXPTIME solvable if we 
restrict the memory of a strategy to be the most recently visited
$t$ vertices where $t$ is linear in the size of the game.

\section{Basic Definitions}

In this section, we define
stochastic multiplayer non-zero-sum game (abbreviated as SMG)
in Sect.~\ref{sec:objectives}
and 
non-cooperative rational synthesis problem (NCRSP) in Sect.~\ref{sec:problems}.
An SMG is a tuple of a game arena (Sect.~\ref{sec:arena})
and winning objectives (Sect.~\ref{sec:objectives}) for all players.
The complexity 
of NCRSP for stationary strategies (Sect.~\ref{sec:strategies})
and other strategy classes 
are analyzed in Sect.~\ref{sec:stationary} and \ref{sec:others}.

\smallskip

First, we define some mathematical notations.
Let $\Nat=\{0, 1, 2, \ldots \}$ be the set of all natural numbers including 0.
For a set $A$,
let $A^{\omega}$ denote the set of all infinite sequences on $A$.
%
A \emph{probability distribution function} on $A$
is a function $\delta: A \to [0,1]$ satisfying $\sum _{a\in A}\delta(a)=1$,
where $[0,1]$ is the set of real numbers between 0 and 1.
The family of all probability distribution functions on $A$ is denoted by $D(A)$.


\subsection{Game Arena}
\label{sec:arena}

A \emph{game arena} (or simply, \emph{arena})
is intuitively
a finite directed graph where both of vertices and edges are labelled.
The label of a vertex expresses which controller controls the vertex.
The label of an edge $(u,v)$ is a rational number $r$ when the source node $u$
is controlled by \emph{nature},
and $r$ represents the transition probability
from the source vertex $u$ to the target vertex $v$.

For any $k\in\Nat$, we can define a ($k+1$) player game arena as 
$\calA=\langle \Omega, V, (V_i)_{i\in\Omega\cup\{\Nature\}},  
\Delta, v_0\rangle $ in which each component is defined as follows.
$\Omega=\{0, 1, \ldots, k\}$ is a finite set of players where 
player 0 corresponds to the \emph{system} and the other players $1$ to $k$
correspond to the \emph{environments}.
Let $\OmegaX{j} = \Omega\setminus{\{j\}}$.
$V$ is a finite set of vertices and $(V_i)_{i\in\Omega\cup\{\Nature\}}$
is a partition of $V$;
i.e, for each vertex $v$, there exists a player or nature $i$
such that $v$ is a member of $V_{i}$, and
for each distinct $i$ and $j$, $V_{i}\cap V_{j}=\emptyset$.
This means that every vertex is controlled by exactly one player or nature. 
If $i\in\Omega$, each member $v$ of $V_i$ is controlled by player $i$ and
the transition probabilities from $v$ to
the next vertices are determined by player $i$'s strategy
defined in Sect.~\ref{sec:strategies}.
If $i=\Nature$, each member $v$ of $V_i$ is controlled by a special controller
$\Nature$, and the transition probabilities from $v$ to the next vertices are
given \emph{a priori}.
Formally, $\Delta \subseteq
 \big(V_{\Nature} \times[0,1]\times V\big) \cup
 \big((\bigcup_{i\in\Omega} V_i)\times\{\bot\}\times V\big)$
is a transition relation.
It is assumed that for each $v\in V_{\Nature}$ and $w\in V$,
there exists exactly one rational number $p\in [0,1]$ such that
$(v,p,w)\in\Delta$.
For each $v\in V$, let
$\Next(v)=\{\, w\in V \mid \exists p \in (0,1]\cup\{\bot\}.$
$(v,p,w)\in \Delta \,\}$.
For each $v\in V$, $\Next(v)$ must not be empty.
We call $v\in V$ \emph{a terminal vertex} if $\Next(v)=\{v\}$;
i.e., $v$ is a terminal vertex if the self-loop $(v,v)$
is the only outgoing edge from $v$.
For each $v\in V_{\Nature}$,
let $\nextPr_v(w)=p$ for each $w\in V$ where
$p\in [0,1]$ is the rational number satisfying $(v,p,w)\in\Delta$.
We require that for every $v\in V_{\Nature}$,
$\nextPr_v$ should be
a probability distribution function on~$V$.
Finally, $v_0\in V$ is the \emph{initial vertex}.

A \emph{play} $\pi = v_0v_1v_2\ldots \in V^{\omega}$
is an infinite sequence
that starts with $v_0$ and
for every $i\in \Nat$, $v_{i+1}\in \Next(v_{i})$.
Let $\pi(i)$ denote $v_i$ for $\pi=v_0v_1v_2\ldots$ and $i\in\Nat$.
Let $\Inf(\pi)$ denote the set of vertices that appear in $\pi$ infinitely often. 



\subsection{Objectives}
\label{sec:objectives}

In this paper, we assume that the result
a player obtains from a play is either a winning or a losing.
Since we are considering non-zero-sum games,
one player's winning does not mean other players' losing.
Each player has her own winning condition over plays,
and we model the condition as a subset $O$ of plays;
i.e., the player wins if the play belongs to the subset $O$.
%
A \emph{winning objective} (or simply, \emph{objective}) of player $i$
is a subset $O_{i}\subseteq V^{\omega}$
of infinite sequences of vertices.
If a play $\pi \in V^{\omega}$ is a member of player $i$'s winning objective $O_i$,
then the \emph{payoff} $\Pay_{i}(\pi)$ of player $i$ on $\pi$ is 1;
otherwise, $\Pay_{i}(\pi)=0$.
The tuple of all players' payoffs on a play $\pi$ is expressed by
$\Pay(\pi)=(\Pay_{0}(\pi), \Pay_{1}(\pi),\ldots ,\Pay_{k}(\pi))$.

In this paper, among major classes of objectives, 
we only deal with terminal reachability and Muller objectives. 
%

\begin{description}[font=\mdseries\itshape,nosep,topsep=\smallskipamount]
\item[Terminal reachability (TR) objective]
  is given by a subset $F$ of terminal vertices as
  $\TR(F)=\{\,\pi\in V^{\omega} \mid \exists n\geq 0$.\ $\pi(n)\in F\,\}$.
\item[Muller objective]
  is given by a Boolean formula $\phi$ over $V$ as
  $\Muller(\phi)=\{\,\pi\in V^{\omega} \mid \Inf(\pi)\models\phi\,\}$,
  where
  $\Inf(\pi)\models\phi$ means that $\phi$ is evaluated to $1$
  under the truth assignment $\theta: V\rightarrow \{0,1\}$
  such that $\theta(v)=1 \Leftrightarrow v\in\Inf(\pi)$.
\end{description}

%

\noindent
The classes of terminal reachability and Muller objectives are abbreviated 
as $\OC{TR}$ and $\OC{Muller}$, respectively.
Major classes of objectives such as $\OC{Buchi}$, $\OC{Parity}$
strictly includes $\OC{TR}$ and are strictly included in $\OC{Muller}$ ~\cite{UW11,CFGR16}. 
We will use $\OC{TR}$ and $\OC{Muller}$ to show 
a lower bound and an upper bound of the complexity of the problem, respectively. 

Note that although we have defined objectives to be Boolean, 
because of the existence of the nature,
in effect it is the same as considering objectives giving
any rational number in $[0,1]$ as a payoff (see \cite[Sect. 2.2]{KTS24}).

\subsection{Games and Strategies}
\label{sec:strategies}

A \emph{stochastic multiplayer non-zero-sum game (SMG)}
is a pair $\calG = \langle\calA, (O_{i})_{i\in\Omega}\rangle$ of a game arena 
$\calA= \langle \Omega, V,\allowbreak (V_i)_{i\in\Omega\cup\{\Nature\}},  \Delta, v_0\rangle$ 
and winning objectives $(O_{i})_{i\in\Omega}$ for all players.
If $V_{\Nature}=\emptyset$, then we call $\calG$
a \emph{non-stochastic multiplayer non-zero-sum game (NSMG)}.
If $|\Omega|=1$ then $\calG$ can be regarded as a Markov decision process (MDP)\@.



A \emph{strategy} of player $i$ is a function $\sigma_i: V^* V_i\rightarrow D(V)$. 
A strategy with memory is formally defined by a Mealy machine \cite{UW11}, 
which we omit here. 
A strategy is {\em finite-state} if it is defined by a Mealy machine whose state set 
is finite.  
%
%
A strategy $\sigma_i$ is {\em memoryless} if
for each $x,x'\in V^{\ast}$ and $v\in V_i$, the condition $\sigma_i(xv)=\sigma_i(x'v)$ holds. 
A memoryless strategy can be represented by a function $\sigma_i: V_i\rightarrow D(V)$. 
A strategy is non-stochastic (or {\em pure}) if 
$\forall i\in\Omega, \forall x\in V^{\ast}, \forall v\in V_{i}, 
\exists w\in V. \sigma_i(xv)(w)=1$ holds. 
A pure strategy can be expressed by a function $\sigma_i: V^* V_i\rightarrow V$. 
By the above definition, the set of strategies can be categorized into six classes:  
unrestricted, pure, finite-state, pure finite-state, memoryless (or {\em stationary}) 
and pure memoryless (or {\em positional}) strategies. 
Restricting the class of strategies to one of the above mentioned six subclasses, 
the problem NCRSP, formally defined later, can also be categorized into six subproblems. 
In Sect.~\ref{sec:stationary} and \ref{sec:others}, 
we also consider some strategy classes other than the above six classes. 

A \emph{strategy profile} is a tuple $\barsigma=(\sigma_i)_{i\in\Omega}$
where $\sigma_i$ is a strategy of player~$i$.
By $\barsigma_{-j}$ we mean the tuple $(\sigma_i)_{i\in\OmegaX{j}}$.
Specifically, $\barsigma_{-0}$ is a tuple of strategies of all environments' strategies.
We call $\barsigma$ (or $\barsigma_{-0}$) positional 
if every player's strategy in $\barsigma$ (or $\barsigma_{-0}$) is positional 
and the same naming convention applies for other strategy classes. 
When $\barsigma=(\sigma_i)_{i\in\Omega}$ is given for all players, 
together with the transition probabilities of nature, 
the transition probability from $u$ to $v$ for every $u,v\in V$ is fixed.
Once the transition probabilities are fixed, 
we can define the probability $\PR^{\barsigma}(O)$ for each $O\subseteq V^{\omega}$ 
\cite[Sect.~2.3.2]{UW11}. 
Intuitively speaking, $\PR^{\barsigma}(O)$ means the probability of executing plays in $O$ 
when the game starts at $v_0$ and follows the strategy profile $\barsigma$. 
The notation $\PR_{v}^{\barsigma}(O)$ has the same meaning as $\PR^{\barsigma}(O)$ 
except that the initial vertex is $v$. 
The \emph{payoff} of player $i$ taken by a strategy profile $\barsigma$ is
$\Pay_i (\barsigma)=\PR^{\barsigma}(O_i)$ where $O_i$ is the winning objective of player $i$. 
The \emph{payoff profile} is the tuple 
$\Pay(\barsigma)=(\Pay_{0}(\barsigma), \Pay_{1}(\barsigma),\ldots ,\Pay_{k}(\barsigma))$
of the payoffs of all players.


\subsection{Nash Equilibria}
\label{sec:nash}

In order to represent the rational behaviors of environment players,
we adopt the concept of Nash equilibrium.
Given a strategy profile $\barsigma$ and a strategy $\tau$ for player $i$, 
we define $(\barsigma_{-i}, \tau)$ as the strategy profile where
player $i$ takes the strategy $\tau$ and the others follow $\barsigma$.

A strategy profile $\barsigma$ is a \emph{Nash equilibrium} if
no player can improve her payoff by changing her strategy alone; i.e.,
$\forall i\in\Omega$, $\forall\tau.$
$\Pay_i (\barsigma_{-i},\tau)\leq \Pay_i (\barsigma)$.
If $\barsigma$ is a Nash equilibrium, we write $\NE(\barsigma)$.
If $\NE(\barsigma)$ does not hold,
then there exists a player $i\in\Omega$ and her alternative strategy $\tau$
such that 
she can gain a better payoff by switching her strategy to $\tau$; i.e.,
$\Pay_i (\barsigma_{-i},\tau)> \Pay_i (\barsigma)$.
In this case,
$\tau$ (or $(\barsigma_{-i},\tau)$) is called
a \emph{profitable deviation} of $i$.
Similarly, a strategy profile $\barsigma$ is a \emph{0-fixed Nash equilibrium} if 
no environment player (not including the system) can improve her payoff
by changing her strategy alone; i.e.,
$\forall i\in\OmegaX{0}$, $\forall\tau.$
$\Pay_i (\barsigma_{-i},\tau)\leq \Pay_i (\barsigma)$.
If $\barsigma$ is a 0-fixed Nash equilibrium, we write $\ZNE(\barsigma)$. 
A profitable deviation is defined for 0-fixed Nash equilibria in the same way.

A 0-fixed Nash equilibrium fits the situation where the system decides her strategy in advance and 
the environment players then decide their strategies subject to rationality restrictions. 
In this paper, we choose Nash equilibria as the rationality constraint
and define NCRSP based on them in the next subsection.


\subsection{Rational Synthesis Problems}
\label{sec:problems}

A rational synthesis problem aims at finding a desirable strategy of the system 
under the assumption that the environment players behave rationally. 
The synthesis problem based on SMG has two variants:
cooperative rational synthesis problem and
non-cooperative rational synthesis problem.

\emph{Cooperative rational synthesis problem (CRSP)} is defined as follows:
Given an SMG $\calG$ and a rational number $\mu\in [0,1]$ that indicates 
the minimum payoff that the system must meet,
is there a Nash equilibrium $\barsigma$ 
where the system's payoff gained from $\barsigma$ is not less than $\mu$?
Formally, CRSP is the decision problem to decide 
$\exists \barsigma.$ $(\NE(\barsigma) \wedge \Pay_0(\barsigma)\ge\mu)$.
In general, more than one Nash equilibrium can exist;
CRSP only asks if there exists at least one Nash equilibrium
satisfying the conditions.

\emph{Non-cooperative rational synthesis problem (NCRSP)}
on the other hand is defined as follows:
Given $\calG$ and $\mu$,
is there a strategy $\sigma_0$ of the system
such that for each 0-fixed Nash equilibrium $\barsigma$ containing $\sigma_0$,
the system's payoff gained from $\barsigma$
always greater than or equal to $\mu$?
Formally, NCRSP is the decision problem to decide 
$\exists \sigma_0$, $\forall \barsigma_{-0}.$
$(\ZNE(\barsigma) \Rightarrow \Pay_0(\barsigma)\ge\mu)$
where $\barsigma=(\sigma_i)_{i\in\Omega}$.

The definition of CRSP fits the situation where 
the system can make a suggestion of a strategy profile to the environment players.
In CRSP, it is assumed that
the environment players follow the suggested strategy profile
if it is a Nash equilibrium.
However, in real applications, the system may rarely have an opportunity to make suggestions to the environment, 
and thus CRSP is optimistic.

In \cite{KTS24}, 
the complexity of NCRSP for positional strategies is analyzed. 
In this paper, we analyze the decidability and complexity of subproblems of NCRSP
for strategies more general than positional ones. 
We provide the naming conventions for subproblems of NCRSP. 
NCRSP with stationary strategies and Muller objectives for all players is denoted by 
Stationary-Muller-NCRSP\@.
For other winning objective classes, the same conventions apply.
We write NNCRSP to denote the NCRSP on NSMG.


\section{Stationary-NCRSP}
\label{sec:stationary}
In this section, 
we first analyze the complexity of Stationary-NCRSP where all players have stationary strategies. 
Next, the complexity of Stationary-Positional-NCRSP where only the system can behave randomly and 
each environment takes a positional strategy is analyzed. 

\subsection{Stationary-NCRSP}
We first give some additional definitions and basic analysis methods
used for proving the main theorem. 


For an arena ${\cal A}$, 
if a strategy profile ${\barsigma_{-i}}$ of all the players other than $i\in \Omega$ is given
and we consider plays according to ${\barsigma_{-i}}$, then 
${\cal A}$ can be regarded as an MDP with a single player $i$. 
For an MDP $\cal{M}$ with vertex set $V$ and a single player $i$, 
we say that $E\subseteq V$ is an {\em end component} of $\cal{M}$ 
if the following three conditions are satisfied: 
(i) $E$ is strongly connected; 
(ii) $\Next(v)\cap E\ne\emptyset$ for each $v\in E$; 
(iii) $\Next(v)\subseteq E$ for each vertex $v\in V\setminus V_i$. 
Let $\EC_{\cal{M}}$ denote the set of all end components of $\cal{M}$.
For a given end component $E$ of an MDP $\cal{M}$, 
we can construct a stationary strategy for $i$ such that if a play starts at $v\in E$, 
the play almost surely visits all and only the vertices in $E$ infinitely often.  
Formally, the next proposition holds. (See \cite[Lemma 2.3]{UW11}.) 

\begin{proposition}\label{EC}
Let ${\cal{M}}$ be an MDP with a player $i$. 
There exists a stationary strategy $\tau$ for player $i$ that satisfies:\\
$\forall E\in \EC_{\cal{M}},\forall v\in E$, 
     $Pr^{\tau}_{v}(\{\pi\in V^{\omega}\mid \Inf(\pi)=E\})=1.$
\end{proposition}

We can decide whether player $i$ has a profitable deviation from $\barsigma$
based on Proposition \ref{EC}. 
(See \cite{BK08} for the theory of end components for MDP.) 
However, the number of possible stationary strategy profiles $\barsigma$ is infinite. 
Instead of an independent strategy profile, 
we will use a support defined below. 
For a player $i$, a set 
$S_i \subseteq \{ (u,v) \mid u\in V_{i}, (u,\bot,v)\in\Delta \}$ 
is a {\em support} for $i$ 
if for each $u\in V_{i}$, at least one $v$ exists such that $(u,v)\in S_i$. 
Intuitively, a support for player $i$ is a subset of edges  
such that there is a strategy of player $i$ assigning 
a positive transition probability to all and only edges in the subset. 
Let $\SUPP^{\cal{A}}_0$ and $\SUPP^{\cal{A}}_{\OmegaX{0}}$ denote 
the set of all supports for the system and 
the set of the unions of supports for the environment players, respectively. 
Also, let $\SUPP^{\cal{A}} = 
\{ X \cup Y \mid X \in \SUPP^{\cal{A}}_0, Y \in \SUPP^{\cal{A}}_{\OmegaX{0}} \}$. 

For an arena ${\cal{A}}$ and $S \in \SUPP^{\cal{A}}$, 
let $S_{-i} = \{ (u,v) \in S \mid u \not\in V_i \}$. 
We define 
the subarena ${\cal{A}}^{S}$ and ${\cal{A}}^{S_{-i}}$ of ${\cal{A}}$  
by restricting edges to $S$ and $S_{-i}$, respectively: 
\begin{quote}
${\cal{A}}^{S}=\langle \Omega, V, (V_i)_{i\in\Omega\cup\{\Nature\}}, 
\Delta\setminus \Delta^{-S}, v_0\rangle $ where 
$\Delta^{-S}=\{(u,\bot,v)\mid (u,\bot,v)\in \Delta, (u,v)\notin S\}$ and \\
${\cal{A}}^{S_{-i}}$ is defined in a similar way. 
%
\end{quote}
\noindent
To analyze the complexity of Stationary-NCRSP, we use the following property.  

\begin{proposition}[\hspace{1sp}\cite{C75}]
\label{EAFL}
The decision problem of the validity of a first-order sentence 
in the theory of real numbers with order ${\cal{R}}:=(R,+,\times, 0, 1,\leq)$ is in 2-EXPTIME. 
\end{proposition}
The next theorem is the main result of this section. 
The design of the first-order formulas in the proof is inspired by \cite[Theorem 4.5]{UW11}. 
\begin{theorem}\label{SMN}
Stationary-Muller-NCRSP is in 3-EXPTIME. 
\end{theorem}
\begin{proof}
Assume that we are given 
an SMG ${\cal{G}} = \langle  {\cal{A}}, \allowbreak {{(O_{i})}_{i\in\Omega}} \rangle$ where 
${\cal A}= \langle \Omega, V, (V_i)_{i\in\Omega\cup\{\Nature\}}, \Delta, v_0\rangle $ 
and a rational number $\mu$. 
We construct a first-order sentence meaning that the answer of 
Stationary-Muller-NCRSP for the given ${\cal{G}}$ and $\mu$ is yes. 
%
A {\em bottom strongly connected component} of $\cal{A}$ is  
a strongly connected component $C$ of $\cal{A}$
such that $\Next(c)\subseteq C$ holds for each $c\in C$.

We first compute $\SUPP^{\cal{A}}_0$, $\SUPP^{\cal{A}}_{\OmegaX{0}}$, $\SUPP^{\cal{A}}$.
Next, for each $S\in\SUPP^{\cal{A}}$, we compute 
the set of all bottom strongly connected components of $\cal{A}^{S}$, 
denoted as $\BSCC_{{\cal{A}}^{S}}$, and 
$\EC_{{\cal{A}}^{S_{-i}}}$ for each $i\in\Omega$. 
Each of the above sets can be computed in PTIME \cite{BK08}. 

Next, for each $S \in \SUPP^{\cal{A}}$ and environment player $i\in\OmegaX{0}$, 
we compute the following three sets. (See the proof of \cite[Theorem 4.5]{UW11}.) 
\begin{quote}
$B^{S}_{i}=\bigcup_{C\in \BSCC_{{\cal{A}}^{S}}, C\models\phi_{i}}{C}$,\\
$R^{S}_{i} = \{ v\in V \mid 
  \mbox{$v$ is reachable to $B^{S}_{i}$ on ${\cal{A}}^{S}$} \}$,\\ 
$E^{S}_{i}=\bigcup_{E\in \EC_{{\cal{A}}^{S_{-i}}}, E\models\phi_{i}}{E}$. 
\end{quote}

Let $\alpha^{0}=(\alpha_{vw})_{v\in V_{0},w\in V}$, 
$\alpha^{-0}=(\alpha_{vw})_{i\in\OmegaX{0},v\in V_{i},w\in V}$, 
$\barz=(z^{i}_{v})_{i\in\Omega,v\in V}$, and
$\barr=(r^{i}_{v})_{i\in\OmegaX{0},v\in V}$ 
be sets of variables. 
Let $\baralpha = \alpha^{0}\cup\alpha^{-0}$. 
Informally, each variable has the following meaning: 
$\alpha_{vw}$ expresses the probability of transiting from vertex $v$ to vertex $w$, 
$z^{i}_{v}$ represents the payoff of player $i$ if the game starts at vertex $v$ 
and $r^{i}_{v}$ represents the maximum payoff of player $i$ 
if the game starts at vertex $v$ and only player $i$ can change her own strategy. 

We construct the following four types of first-order formulas. 
\begin{eqnarray*}
\lefteqn{
\psisys_{S_0}(\alpha^0) := \bigwedge_{v\in V_0}(\bigwedge_{w\in \Next(v)}\alpha_{vw}\ge 0\land 
}\\
& & \bigwedge_{w\in V\setminus \Next(v)}\alpha_{vw}=0 \land \sum_{w\in \Next(v)}\alpha_{vw}=1)\land\\ 
& & \bigwedge_{(v,w)\in S_0}\alpha_{vw}>0 \land \bigwedge_{(v,w)\notin S_0,v\in V_0}\alpha_{vw}=0. 
\end{eqnarray*}
\noindent
(Variables $\alpha_{vw} \in \alpha^0$ have real values as probabilities and 
are consistent with support $S_0$ for the system.) 
\begin{eqnarray*}
\lefteqn{ \psienv_{S_{-0}}(\alpha^{-0}) }\\
&:=& \bigwedge_{v\in V_i, i\in\OmegaX{0}}(\bigwedge_{w\in \Next(v)}\alpha_{vw}\ge 0\land\\ 
&  & \bigwedge_{w\in V\setminus \Next(v)}\alpha_{vw}=0 \land\sum_{w\in \Next(v)}\alpha_{vw}=1)\land\\
&  & \bigwedge_{v\in V_{\Nature},(v,p,w)\in \Delta}\alpha_{vw}=p\land
     \bigwedge_{(v,w)\in S_{-0}} \alpha_{vw}>0 \land \\
&  & \bigwedge_{(v,w)\notin S_{-0},v\in V_j,j\in\OmegaX{0}} \alpha_{vw}=0. 
\end{eqnarray*}
\noindent
(Similarly to the first formula, 
$\alpha_{vw} \in \alpha^{-0}$ have real values as probabilities and 
are consistent with support $S_{-0}$ for the environment players.) 
\medskip\par\noindent
For each $i\in\Omega$, 
\begin{eqnarray*}
\lefteqn{ \psi^z_{i,S}(\baralpha,\barz) }\\
&:=& \bigwedge_{v\in B^{S}_{i}}z^i_v =1\land \bigwedge_{v\in V\setminus R^S_i}z^i_v =0\land\\
&  & \bigwedge_{v\in R^S_i\setminus B^S_i}z^i_v =\sum_{w\in \Next(v)}\alpha_{vw}\times z^i_w. 
\end{eqnarray*}
\noindent
(If $v$ is in a bottom strongly connected component that satisfies the objective of player $i$ 
then $z^i_v$ is 1;  
else if $v$ is reachable to such a component 
     then $z^i_v$ is determined by the transition probabilities to 
          the next vertices and their payoffs; 
else $z^i_v$ is 0.) 
\medskip\par\noindent
For each $i\in\OmegaX{0}$, 
\begin{eqnarray*}
\lefteqn{ \psi^r_{i,S}(\baralpha,\barr) }\\
&:=& \bigwedge_{v\in V}r^i_v \ge 0 \land \bigwedge_{v\in E^{S}_{i}}r^i_v =1\land 
     \bigwedge_{v\in V_i ,w\in \Next(v)}r^i_v \ge r^i_w \land \\
&  & \bigwedge_{v\in V\setminus V_i}r^i_v =\sum_{w\in \Next(v)}\alpha_{vw}\times r^i_w. 
\end{eqnarray*}
\noindent
(This formula is similar to the third one 
except that only player $i$ can alter her strategy, which can be described by 
using end components instead of bottom strongly connected components.
See Proposition \ref{EC} for the correctness of the second term.) 
%
\medskip\par\noindent
Finally, we construct the following first-order sentence $\psi$. 
We can easily see that $\psi$ is true if and only if 
the answer of Stationary-Muller-NCRSP for ${\cal{G}}$ and $\mu$ is yes. 
\begin{align*}
\psi 
=& \bigvee_{S_{0}\in \SUPP^{\calA}_{0}} 
           \biggl( \exists \alpha^{0}\!.\ \psisys_{S_0}(\alpha^{0}) \land{}\\
 & \hspace{-0.3em}
   \bigwedge_{S_{-0}\in \SUPP^{\calA}_{\OmegaX{0}}}
      \Bigl( \forall \alpha^{-0}\, \forall \barz\, \forall \barr.\
    \bigl( \psienv_{S_{-0}}(\alpha^{-0})\land{}\\
 & \hspace{2em}
   \bigwedge_{i\in\Omega} \psi^{z}_{i,S_0\cup S_{-0}}(\baralpha,\barz)\land
    \bigwedge_{i\in\OmegaX{0}} \psi^{r}_{i,S_0\cup S_{-0}}(\baralpha,\barr)\bigr) \\
 & \hspace{1.5em}
   \implies
 \bigl(\bigwedge_{i\in\OmegaX{0}} z^i_{v_0}\ge r^i_{v_0} \implies z^0_{v_0}\ge\mu\bigr)\Bigr)\biggr)
\end{align*}
\noindent
(The subformula in the last line states that 
if no environment player $i$ can increase her payoff at $v_0$ by unilaterally changing her strategy, 
the payoff of the system at $v_0$ is at least $\mu$.)
\medskip\par\noindent
Sentence $\psi$ can be constructed in exponential time. 
By Proposition \ref{EAFL}, the validity of $\psi$ can be decided in 2-EXPTIME
in the size of $\psi$. 
Thus, this algorithm works in 3-EXPTIME. 
\end{proof}


\subsection{Stationary-Positional-NCRSP}

In this section, the complexity of NCRSP where only the system can take a stationary strategy and 
the environment players must take positional strategies is analyzed. 
This subproblem of NCRSP is called Stationary-Positional-NCRSP\@. 
We show Stationary-Positional-NCRSP is in NEXPSPACE. 


Stationary-Positional-NCRSP is a distinct problem from Pure-NCRSP and Stationary-NCRSP. 
First, we show an example for which the answer of Stationary-Positional-NCRSP is yes 
but the answer of Pure-NCRSP is no. 
\begin{figure}[tb]
\begin{center}
\includegraphics[width=8cm]{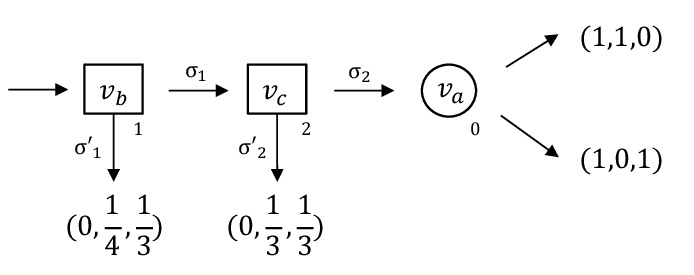}
\end{center}
\caption{An SMG for which the answer of Stationary-Positional-NCRSP is yes and that of Pure-NCRSP is no with $\mu =1$}
\label{fig:ex1}
\end{figure}
\begin{example}
Figure \ref{fig:ex1} shows an SMG played by three players $0,1,2$ who have TR objectives. 
Let $\mu =1$. 
Let $\sigma_{1}$ and $\sigma_{1}'$ be the strategies of player 1 that transits 
from $v_{b}$ to $v_{c}$ and from $v_{b}$ to $(0,\frac{1}{4},\frac{1}{3})$, respectively.  
Similarly, let $\sigma_{2}$ and $\sigma_{2}'$ be the strategies of player 2 that transits 
from $v_{c}$ to $v_{a}$ and from $v_{c}$ to $(0,\frac{1}{3},\frac{1}{3})$, respectively.   

For the SMG, the answer of Stationary-Positional-NCRSP is yes. 
Indeed, consider the strategy $\sigma_{0}$ for player $0$ that transits from $v_a$ to $(1,1,0)$ 
with probability $\frac{1}{2}$ and from $v_a$ to $(1,0,1)$ with probability $\frac{1}{2}$. 
Then, for each strategy profile $\barsigma$, $Pay_{0}(\barsigma)\ge 1 \lor \neg 0NE(\barsigma)$ holds. 
Table \ref{tb:ex1} shows each strategy profile $\barsigma$ and 
the payoff $Pay(\barsigma)$ taken by the strategy profile with $\sigma_{0}$ fixed.  

\begin{table}[tb] \label{tb:ex1}
\begin{center}
\caption{All strategy profiles $Pay(\barsigma)$ with $\sigma_{0}$ fixed. }
\begin{tabular}{|c||c|c|}
\hline
strategies for player 1$\backslash$for player 2 &  $\sigma_{2}$ &  $\sigma_{2}'$ \\ \hline \hline
$\sigma_{1}$                 & $(1,\frac{1}{2},\frac{1}{2})$ & $(0,\frac{1}{3},\frac{1}{3})$\\ \hline
$\sigma_{1}'$                & $(0,\frac{1}{4},\frac{1}{3})$ & $(0,\frac{1}{4},\frac{1}{3})$ \\ \hline
\end{tabular}
\end{center}
\end{table}

The strategy profile $\barsigma$ that satisfies $0NE(\barsigma)$ is 
only $\barsigma=(\sigma_{0}, \sigma_{1} ,\sigma_{2})$. 
Therefore, the answer of Stationary-Positional-NCRSP for the SMG with $\mu =1$ is yes. 

Next, we show that the answer of Pure-NCRSP for the SMG with $\mu =1$ is no. 
Because of the restriction on pure strategies, 
player $0$ must deterministically transit from $v_A$ to one of $(1,1,0)$ and $(1,0,1)$. 
If player $0$ chooses $(1,1,0)$, the strategy profile $\barsigma=(\sigma_{0}, \sigma_{1} ,\sigma_{2}')$ 
satisfies $0NE(\barsigma)$ and player $0$ loses in $\barsigma$. 
Similarly, if player $0$ chooses $(1,0,1)$, 
the strategy profile $\barsigma'=(\sigma_{0}, \sigma_{1}' ,\sigma_{2})$ 
satisfies $0NE(\barsigma)$ and player $0$ loses in $\barsigma'$. 
Hence,  the answer of Pure-NCRSP for the SMG with $\mu =1$ is no. 
\qed
\end{example}

Next, an example for which the answer of Pure-NCRSP is yes 
but the answer of Stationary-Positional-NCRSP is no is shown. 

\begin{figure}[tb]
\begin{center}
\includegraphics[width=6cm]{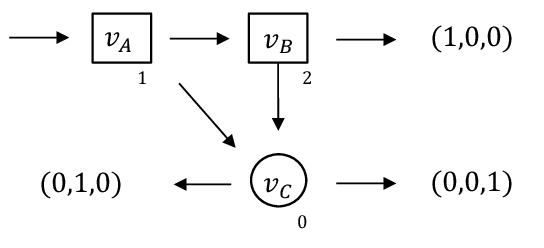}
\end{center}
\caption{An SMG for which the answer of Pure-NCRSP is yes 
and that of Stationary-Positional-NCRSP is no with $\mu =1$}
\label{fig:ex2}
\end{figure}

\begin{example}
The SMG in Figure \ref{fig:ex2}
was first given in \cite[Figure 4]{UW11} as an example for which the answer of Pure-NCRSP is yes 
but the answer of Stationary-NCRSP is no. 
It is easily proved that the answer of Stationary-Positional-NCRSP for the SMG is no. 
\qed
\end{example}

Next, we show that Stationary-Positional-NCRSP and Stationary-NCRSP are distinct. 
For the following SMG, the answer of Stationary-NCRSP is yes 
but the answer of Stationary-Positional-NCRSP is no. 

\begin{figure}[tb]
\begin{center}
\includegraphics[width=5cm]{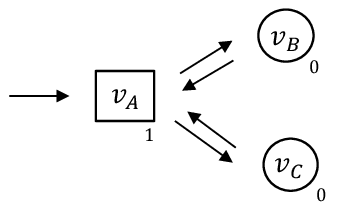}
\end{center}
\caption{An SMG for which the answer of Stationary-NCRSP is yes and 
that of Stationary-Positional-NCRSP is no with $\mu =1$}
\label{fig:ex3}
\end{figure}

\begin{example}
Figure \ref{fig:ex3} illustrates an SMG played by two players $0,1$. 
We assume the winning objectives for players $0,1$ are both 
$\Muller(\phi)$ with $\phi = v_{a} \wedge v_{b} \wedge v_{c}$. 
Let $\mu =1$. 
First, we show that the answer of Stationary-NCRSP for the SMG with $\mu =1$ is yes. 
Let $\sigma_{1}$, $\sigma_{1}'$ and $\sigma_{1}''$ be the strategies of player 1 
that transits from $v_{a}$ to $v_{b}$ deterministically, from $v_{a}$ to $v_{c}$ deterministically 
and branching from $v_{a}$ to $v_{b}$ and $v_{c}$ with arbitrary positive probability, respectively. 
If player $1$ follows strategy $\sigma_{1}$ or $\sigma_{1}'$, the payoffs of player $0$ and $1$ are 0. 
If player $1$ follows strategy $\sigma_{1}''$, the payoffs of players $0$ and $1$ are 1. 
Because $\sigma_{1}''$ is the only strategy profile that is a 0-fixed Nash equilibrium, 
the answer of Stationary-NCRSP for the SMG with $\mu =1$ is yes. 

Next, we show that the answer of Stationary-Positional-NCRSP for the SMG with $\mu =1$ is no.    
Because the environment player $1$ must follow a positional strategy, 
player $1$ must transit deterministically from $v_{a}$ to $v_{b}$ or $v_{c}$. 
Each of the two strategy profiles is a 0-fixed Nash equilibrium and player $0$ gets payoff $0$. 
Therefore, the answer of Stationary-Positional-NCRSP for the SMG with $\mu =1$ is no. 
\qed
\end{example}

Finally, we show an example for which the answer of Stationary-Positional-NCRSP is yes 
but the answer of Stationary-NCRSP is no. 

\begin{figure}[tb]
\begin{center}
\includegraphics[width=8cm]{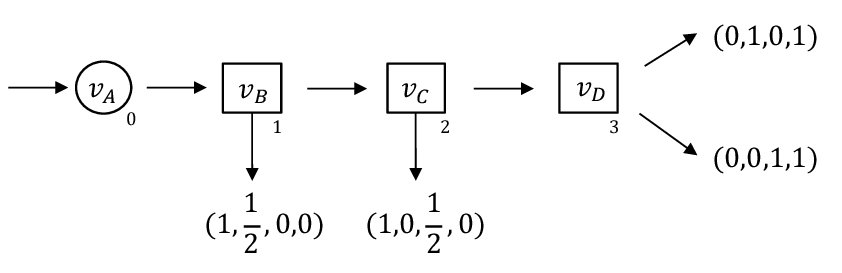}
\end{center}
\caption{An SMG for which the answer of Stationary-Positional-NCRSP is yes 
and that of Stationary-NCRSP is no with $\mu =1$}
\label{fig:ex4}
\end{figure}
\begin{example}
Figure \ref{fig:ex4} shows an SMG played by four players $0,1,2,3$ and each player has a TR objective. 
Let $\mu =1$. 
We show that the answer of Stationary-Positional-NCRSP for the SMG with $\mu =1$ is yes. 
Because every environment player must follow a positional strategy,  
player $3$ must transit deterministically from $v_{D}$ to $(0,1,0,1)$ or $(0,0,1,1)$. 
If player $3$ chooses $(0,1,0,1)$, each strategy profile that reaches $(0,1,0,1)$ 
cannot be a 0-fixed Nash equilibrium because player $2$ at $v_C$ has a profitable deviation. 
Similarly, if player $3$ chooses $(0,0,1,1)$, each strategy profile that reaches $(0,0,1,1)$ 
cannot be a 0-fixed Nash equilibrium because player $1$ at $v_B$ has a profitable deviation. 
This implies that every 0-fixed Nash equilibrium must 
reach $(1,\frac{1}{2},0,0)$ or $(1,0,\frac{1}{2},0)$ and hence 
the answer of Stationary-Positional-NCRSP for the SMG with $\mu =1$ is yes. 

Next we show that the answer of Stationary-NCRSP for the SMG with $\mu =1$ is no. 
Consider strategy $\sigma_{3}$ for player $3$ that branches at $v_{D}$ to $(0,1,0,1)$ and $(0,0,1,1)$, 
each with probability $\frac{1}{2}$. 
Then, consider the strategy profile that reaches $v_{D}$ with probability $1$ 
and player $3$ follows $\sigma_{3}$. 
The strategy profile is a 0-fixed Nash equilibrium where player $0$ loses. 
Hence, the answer of Stationary-NCRSP for the SMG with $\mu =1$ is no. 
\qed
\end{example}


In the following, 
we show that Stationary-Positional-Muller-NCRSP is in EXPSPACE (Theorem \ref{SPMN})
with the help of Proposition \ref{EFL}. 

\begin{proposition}[\hspace{1sp}\cite{JC88}]\label{EFL}
The decision problem of the validity of a first-order sentence that has no universal quantifier 
in the theory of real numbers with order ${\cal{R}}:=(R,+,\times, 0, 1,\leq)$ is in PSPACE.  
\end{proposition}

\begin{theorem}\label{SPMN}
Stationary-Positional-Muller-NCRSP is in NEXPSPACE. 
\end{theorem}
\begin{proof}
We show a nondeterministic exponential space algorithm that solves the problem. 
First, choose a support $S_0\in \SUPP^{\cal{A}}_0$ for the system non-deterministically. 
Next, for each tuple of positional strategies $\barsigma_{-0}=(\sigma_1,\sigma_2,\ldots ,\sigma_k)$ 
of the environment players, 
construct the new game arena ${\calA}^{S_0,\barsigma_{-0}}=\langle \Omega, V, 
(V_i)_{i\in\Omega\cup\{\Nature\}},  
\{(v,\bot,w) \in \Delta \mid (v,w)\in S_0 \} \cup
\{(v,\bot,\sigma_i(v)) \mid v\in V_i,~ i\in\OmegaX{0} \} \cup
\{(v,p,w)\in \Delta \mid v\in V_{\Nature}, p\in (0,1] \, \},  
v_0\rangle $. 
${\calA}^{S_0,\barsigma_{-0}}$ is the $(k+1)$ player game obtained from ${\cal A}$ by 
pruning the edges not chosen by $S_0$ or $\barsigma_{-0}$.  
Remember $\alpha^{0}=(\alpha_{vw})_{v\in V_{0},w\in V}$ used in the proof of Theorem \ref{SMN}. 
By regarding variables $\alpha^{0}$ as probabilities, 
for each ${\cal{A}}^{S_0,\barsigma_{-0}}$, we can compute $Pay(\barsigma_{-0})$ 
in PTIME by the method in the proof of \cite[Theorem 5]{KTS24}. 
Because the number of strategies for the environment players is exponential, 
computing and writing $Pay(\barsigma_{-0})$ for all ${\cal{A}}^{S_0,\barsigma_{-0}}$ 
need exponential space. 
Note that $Pay(\barsigma_{-0})$ is an expression consisting of 
+, $\times$, $\alpha^{0}$ and rational numbers $p$ such that $(v,p,w)\in\Delta$ for some $v$ and $w$.

Next, let $\Sigma_{-0}$ be the set of
the strategy profiles of the environment players
and
$\Deltapd = \{
  (\barsigma_{-0},i,\barsigma'_{-0})
  \in\Sigma_{-0}\times\Omega_{-0}\times\Sigma_{-0} \mid
  \sigma_j=\sigma'_j$ for $j\in\Omega_{-0}\setminus\{i\}$
  and $\sigma_i\ne\sigma'_i\}$.
Each member
$(\barsigma_{-0},i,\barsigma'_{-0})$ of $\Deltapd$
represents a candidate of profitable deviation of player~$i$
from $\barsigma_{-0}$ to~$\barsigma'_{-0}$.
Then, non-deterministically choose a subset
$\Deltapd' \subseteq \Deltapd$,
and construct
a set $\COND$ of first-order formulas as
$\COND =
  \{ \Pay_i(\barsigma_{-0}) < \Pay_i(\barsigma'_{-0}) \mid
     (\barsigma_{-0},i,\barsigma'_{-0})\in\Deltapd' \} \cup
  \{ \Pay_0(\barsigma_{-0}) \ge \mu \mid
     (\barsigma_{-0},i,s)\notin\Deltapd'$ for any $i$ and $s\}$.
$\COND$ gives the constraints that make
the elements of $\Deltapd'$ actual profitable deviations;
for each
$(\barsigma_{-0},i,\barsigma'_{-0})\in\Deltapd'$,
the payoff of $i$ obtained from $\barsigma'_{-0}$ is
greater than that from~$\barsigma_{-0}$.
Moreover, $\barsigma_{-0}$ such that
$(\barsigma_{-0},i,s)\notin\Deltapd'$ for any $i$ and $s$
is a 0-fixed Nash equilibrium and thus
it should satisfy $\Pay_0(\barsigma_{-0}) \ge \mu$.
Finally, by using $\COND$ and variables $\alpha^0$,
construct first-order sentence $\psi$
without universal quantifier as
$\psi = \exists\alpha^0.\ \psisys_{S_0}(\alpha^0)
 \land \bigwedge_{C\in\COND} C$.
By Proposition \ref{EFL}, we can decide whether $\psi$ is valid
in PSPACE\@.
Therefore, the algorithm works in NEXP\-SPACE\@. 
\end{proof}

\section{NCRSP in Other Settings}
\label{sec:others}


\subsection{NCRSP$_{>}$}
In \cite[Theorems 4.9, 4.14]{UW11}, it is shown that TR-CRSP for 
pure finite-state, finite-state, pure and unrestricted strategies are all undecidable. 
In this section, the undecidability of TR-NCRSP$_{>}$ for 
these four subclasses of strategies are shown by reducing TR-CRSP to TR-NCRSP$_{>}$
where NCRSP$_{>}$ is defined as follows: 
Given an SMG $\cal{G}$ and a rational number $\mu \in [0,1]$, decide whether 
%
$\exists \sigma_0.\, \forall \barsigma_{-0}.\,
(\ZNE(\barsigma) \Rightarrow \Pay_0(\barsigma)> \mu)$ 
where $\barsigma=(\sigma_i)_{i\in\Omega}$. 
%
The difference between NCRSP and NCRSP$_{>}$ is the conditions the system should satisfy: 
$Pay_0(\barsigma)\ge \mu$ in NCRSP while $Pay_0(\barsigma)> \mu$ in NCRSP$_{>}$. 
The complement problem of NCRSP$_{>}$ is the decision problem $\overline{{\rm NCRSP_{>}}}$: 
Given an SMG $\cal{G}$ and a rational number $\mu \in [0,1]$, decide whether
%
$\forall \sigma_0.\, \exists \barsigma_{-0}.\,
(\ZNE(\barsigma) \land \Pay_0(\barsigma)\leq\mu)$ 
where
$\barsigma=(\sigma_i)_{i\in\Omega}$.  

\begin{lemma}
Suppose we are given objective classes $\textbf{F}$, $\textbf{F}'$ and a strategy class $S$. 
If $\textbf{F}'\supseteq \textbf{F}\cup co\mathchar`-\textbf{F}$, 
the problem $S$-$\textbf{F}$-CRSP with $k+1$ players can be reduced to 
S-$\textbf{F}'$-$\overline{\text{NCRSP}_{>}}$ with $k+2$ players in polynomial time. 
\end{lemma}
\begin{proof}
CRSP is the decision problem to decide whether 
$\exists \barsigma=(\sigma_i)_{i\in\Omega}.\ ( \NE(\barsigma) \wedge \Pay_0(\barsigma)\ge\mu )$ holds 
given a SMG ${\cal{G}} = \langle  {\cal{A}}, {{(O_{i})}_{i\in\Omega}} \rangle$ where  
${\cal A}= \langle \Omega, V, (V_i)_{i\in\Omega\cup\{\Nature\}},  \Delta, v_0\rangle$ 
and a threshold $\mu\in [0,1]$. 
From ${\calG}$, we construct
${\calG}' = \langle  {\calA}', {{(O_{i})}_{i\in\Omega '}} \rangle$ where 
${\cal A}' = \langle \Omega ', V, (V_i)_{i\in\Omega'\cup\{\Nature\}},\allowbreak
\Delta, v_0\rangle$, 
$\Omega '=\Omega\cup \{0'\}$, player $0'$ is the system of ${\calG}'$, 
$V_{0'}=\emptyset$, and $O_{0'}=\overline{O_0}$. 
Note that we use the complement of the objective for the system in ${\calG}$. 
We reduce the given CRSP instance 
to $\overline{{\rm NCRSP_{>}}}$ for ${\calG}'$ and $1-\mu$, 
which is the problem to decide whether
$\forall \sigma_{0'}.\, \exists (\sigma_i)_{i\in\Omega}.\,
   (\ZpNE(\barsigma') \land \Pay_{0'}(\barsigma')\leq1-\mu)$ 
   where $\barsigma'=(\sigma_i)_{i\in\Omega'}$. 
Because player $0'$ controls no vertex, 
$\forall \sigma_{0'}.\, \exists (\sigma_i)_{i\in\Omega}$ 
is equivalent to $\exists (\sigma_i)_{i\in\Omega}$ 
and $\ZpNE(\barsigma')$ coincides with $\NE(\barsigma)$. 
The reduction is correct 
because $\Pay_{0'}(\barsigma')+\Pay_{0}(\barsigma)=1$, 
which is implied by $O_{0'}=\overline{O_0}$.  
\end{proof}
Though $\textbf{TR}$ is not closed under the complement in general, for each SMG ${\cal{G}}$ 
such that each play of ${\cal{G}}$ almost surely reaches a terminal vertex, 
the complement of any TR objective on ${\cal{G}}$ is also a TR objective. 
\begin{lemma}
Let ${\cal{G}}$ be an SMG. 
If the arena of ${\cal{G}}$ satisfies the following condition, 
the complement of any TR objective on ${\cal{G}}$ 
can be represented by a TR objective on this ${\cal{G}}$: 
\begin{align*}
\forall \barsigma.\ Pr^{\barsigma}(\{\pi\in V^{\omega}\mid{}&\exists t\in V,\ \Next(t)=\{t\},\\
&\exists n \in \mathbb{N}.\ \pi (n)=t\})=1.
\end{align*}
\end{lemma}

\noindent
By the two lemmas above and the result shown in \cite[Theorems 4.9, 4.14]{UW11}, the following theorem holds. 

\begin{theorem}\label{GTNL}
For pure finite-state, finite-state, pure and unrestricted strategies, 
TR-NCRSP$_{>}$ is undecidable. 
\end{theorem}


\subsection{$t$-memory-NCRSP}

Though it is open whether Pure Finite-State-TR-NCRSP is decidable or not, 
as proved in Theorem \ref{GTNL}, Pure Finite-State-TR-NCRSP$_{>}$ is undecidable. 
However, if we consider applications in real life, it is reasonable to assume 
there is $t \in \Nat$ such that players remember 
at most $t$ vertices that precede the current vertex in the history, 
instead of having finite-memory. 
In this section, NCRSP with $t$-memory is considered. 
We say that a strategy $\sigma_{i}: V^{\ast}V_i\rightarrow D(V)$ is $t$-{\em memory} 
with some $t \in \Nat$
if for $x \in V^{\ast}$ and $v \in V_i$ such that $|x|>t$, we have
$\sigma_i(xv)=\sigma_i(zv)$ where $z$ is the suffix of $x$ such that $|z|=t$. 
A $t$-memory strategy can be given by a function $\sigma_{i}: V^{\le t}V_i\rightarrow D(V)$
with $V^{\le t} = \{ u\in V^{\ast} \mid |u| \le t \}$.  

\begin{theorem}
Given $t\in O(|V|)$, Pure $t$-memory-Muller-NCRSP is in EXPSPACE. 
\end{theorem}
\begin{proof}
Assume that we are given an SMG ${\cal{G}} = 
\langle  {\cal{A}}, {{(O_{i})}_{i\in\Omega}} \rangle$ where 
$\calA=\langle \Omega, V, (V_i)_{i\in\Omega\cup\{\Nature\}},  
\Delta, v_0\rangle $ and a threshold $\mu\in [0,1]$. 
From ${\cal{G}}$, we construct
${\cal{G}}'= \langle  {\cal{A'}}, {{(O_{i}')}_{i\in\Omega}} \rangle$ 
by unfolding ${\cal{A}}$ as follows. 
$\calA'=\langle \Omega, V',\allowbreak (V_{i}')_{i\in\Omega\cup\{\Nature\}}, \Delta', v_0\rangle $ 
where $V'=\{ u \mid u\in V^*$, $|u|\leq t+1 \}$, 
and for each $i\in\Omega\cup\{\Nature\}$, 
$V_{i}'=\{xv\mid x\in V^*$, $|x|\leq t$, $v\in V_{i} \}$.

If $\exists p \in [0,1]\cup \{ \bot \}.(v,p,w) \in \Delta$, then 
let $(xv,p,xvw) \in \Delta '$for every $x$ $(|x|<t)$ and 
let $(xv,p,x[1:]vw) \in \Delta '$for every $x$ $(|x|= t)$. 
The number of vertices of the unfolded game is $O(n^n)$ 
where $n$ is the number of vertices of the original game. 
Each vertex of the unfolded game corresponds to a memory-state of the original game 
that has the suffix of length at most $t$ of the history that precedes the current vertex. 
Let $\phi_{i}'$ be the boolean formula obtained by replacing 
each variable $v_j$ in $\phi_i$ into $x_{0}v_j\lor x_{1}v_j\lor \ldots$ 
with each history $x_0, x_1,\ldots (|x_{k}|\leq t)$ 
for each player $i$'s objective 
$\Muller(\phi_i)=\{\pi\in V^{\omega} \mid \Inf(\pi)\models\phi_i\}$. 
The description size of $\phi_{i}'$ is $O(n^n)$ where $n$ is the size of the original objective. 
Because a pure $t$-memory strategy 
in ${\calG}$ corresponds to a positional strategy in ${\calG}'$, 
Pure $t$-memory-Muller-NCRSP for $\calG$ is reduced to Positional-Muller-NCRSP 
for ${\calG}'$ with exponential grow-up. 
We then apply the algorithm given in the proof in \cite[Theorem 5]{KTS24} 
which runs in PSPACE\@.
Hence, Pure $t$-memory-Muller-NCRSP for a given constant $t\in O(|V|)$
is in EXPSPACE.
\end{proof}

\begin{corollary}
Given a constant natural number $t\in \mathbb{N}$, 
Pure $t$-memory-Muller-NCRSP is in PSPACE. 
\end{corollary}

\noindent
In \cite[Theorem 7]{CFGR16}, it is shown that Pure-TR-NNCRSP is PSPACE-hard. 
By using the same method, we can show that 
QBF with $m$ variables can be reduced to Pure-$t$-memory-TR-NNCRSP where $t= 2m+1$. 
Hence, the following theorem is obtained. 

\begin{theorem}
Given $t\in O(|V|)$, Pure $t$-memory-TR-NNCRSP is PSPACE-hard. 
\end{theorem}

\section{Conclusion}
Computational complexity of NCRSP 
on stochastic multi-player non-zero-sum games was analyzed in this paper. 
In Sect. 3, Stationary-NCRSP and Stationary-Positional-NCRSP were shown 
to be solvable in 3-EXPTIME and NEXPSPACE, respectively. 
In Sect. 4, TR-NCRSP$_{>}$ was shown to be undecidable
for pure finite-state, finite-state, pure and unrestricted strategies
while Pure $t$-memory-Muller-NCRSP was shown to be EXPSPACE solvable. 
It is open whether TR-NCRSP for the above four strategy classes are decidable. 

In this paper, we make some assumptions on SMG. 
A game is assumed to be turn-based 
and the rationality of the environment players can be represented by Nash equilibria. 
A game is called concurrent if 
each vertex of the game arena can be controlled by multiple players. 
Examples of refinement of Nash equilibria 
include subgame perfect equilibria, Doomsday equilibria and robust equilibria. 
Extending this research by considering concurrent games and refined equilibrium concepts 
is left as future work. 
\medskip\par\noindent
{\bf Acknowledgement} The authors would like to thank Dr. Mizuhito Ogawa
for informing them of cylindrical algebraic decomposition and Proposition \ref{EAFL}.



\end{document}